%% file: RR-6838.tex
\documentclass[11pt]{article}
\usepackage{RR}

\usepackage[ruled]{algorithm}
\usepackage{algpseudocode}
\usepackage{graphicx}
\usepackage{amssymb}
\usepackage[isolatin]{inputenc}
\usepackage[OT1]{fontenc}   
\usepackage{dsfont}
\usepackage{epsfig}
\usepackage{url}
\usepackage{ifthen}
\usepackage{latexsym}
\usepackage{array}
\usepackage{float}

\newtheorem{definition}{Definition}
\newtheorem{theorem}{Theorem}
\newtheorem{lemma}{Lemma}

\newtheorem{corollary}{Corollary}

\newenvironment{proof}{{\bf Proof. } }{{\hfill $\Box$}\vspace{.5pc}}

\newcommand{\boldsymbol}[1]{\mbox{\boldmath $#1$}}

\newcommand{\Rem}[1]{\qquad $\slash *$ {\em #1} $* \slash$}

\alglanguage{pseudocode}

\newcommand{\MOVE}{\textbf{Move} }
\newcommand{\TRY}{\textbf{Try to move} }

\RRdate{February 2009}
\RRtitle{Exploration d'Anneau Probabiliste Optimale par des Robots Asynchrones et Amnésiques}
\RRetitle{Optimal Probabilistic Ring Exploration by Asynchronous Oblivious Robots}
\titlehead{Optimal Probabilistic Ring Exploration by Asynchronous Oblivious Robots}
\RRauthor{St\'{e}phane Devismes$^\circ$ \and Franck Petit$^\dag$ \and S\'{e}bastien Tixeuil$^{\star,\ddag}$\\
{\small
$^\circ$ VERIMAG UMR 5104, Universit\'e Joseph Fourier, Grenoble (France)\\
$^\dag$ INRIA, LIP UMR 5668, Université de Lyon / ENS Lyon (France)\\
$^\star$ Universit\'{e} Pierre et Marie Curie - Paris 6, LIP6, France\\
$^\ddag$ INRIA project-team Grand Large
}
}

\authorhead{S. Devismes \emph{et al.}}

\RRresume{Nous considérons une équipe de $k$ robots identiques, amnésiques, asynchrones et mobiles qui sont capables de percevoir leur environnement mais incapables de communiquer, et évoluent sur des circuits contraints. Les résultats précédents qui utilisent le même scénario montrent que la symmétrie initiale potentielle induit des bornes inférieures élevées dès lors que le problème doit être résolu par des robots déterministes.

Dans cet article, nous initions la recherche sur les bornes et sur les solutions probabilistes dans le même contexte, et nous considérons le problème de l'exploration d'anneaux anonymes et non orientés de taille quelconque. Il est connu que $\Theta(\log n)$ robots sont nécessaires et suffisants dans le cas déterministe pour résoudre le problème avec $k$ robots, tant que $k$ et $n$ sont premiers entre eux. En contrepartie, nous montrons que quatre robots identiques probabilistes sont nécessaires et suffisants pour résoudre le même problème, tout en supprimant la contrainte de coprimalité. Nos résultats positifs sont constructifs.
}

\RRabstract{We consider a team of $k$ identical, oblivious, asynchronous mobile robots that are able to sense (\emph{i.e.}, view) their environment, yet are unable to communicate, and evolve on a constrained path. Previous results in this weak scenario show that initial symmetry yields high lower bounds when problems are to be solved by \emph{deterministic} robots.

In this paper, we initiate research on probabilistic bounds and solutions in this context, and focus on the \emph{exploration} problem of anonymous unoriented rings of any size. It is known that $\Theta(\log n)$ robots are necessary and sufficient to solve the problem with $k$ deterministic robots, provided that $k$ and $n$ are coprime. By contrast, we show that \emph{four} identical probabilistic robots are necessary and sufficient to solve the same problem, also removing the coprime constraint. Our positive results are constructive.
}

\RRmotcle{Robots, Anonymat, Amnésie, Exploration
}
\RRkeyword{Robots, Anonimity, Obliviousness, Exploration
}
\RRprojet{Grand Large}
\RRtheme{\THNum}
\URFuturs

\begin{document}
\RRNo{6838}
\makeRR

\input{intro}

\input{model}

\input{imp}

\input{algo}

\input{ccl}

\bibliographystyle{plain} 
\bibliography{../../../biblio/biblio}

\end{document}

%% file: intro.tex
\section{Introduction}

We consider autonomous robots that are endowed with visibility sensors (but that are otherwise unable to communicate) and motion actuators. Those robots must collaborate to solve a collective task, namely \emph{exploration}, despite being limited with respect to input from the environment, asymetry, memory, etc. In this context, the exploration tasks requires every possible location to be visited by at least one robot, with the additional constraint that all robots stop moving after task completion.

Robots operate in \emph{cycles} that comprise \emph{look}, \emph{compute}, and \emph{move} phases. The look phase consists in taking a snapshot of the other robots positions using its visibility sensors. In the compute phase a robot computes a target destination based on the previous observation. The move phase simply consists in moving toward the computed destination using motion actuators. 

The robots that we consider here have weak capacities: they are \emph{anonymous} (they execute the same protocol and have no mean to distinguish themselves from the others), \emph{oblivious} (they have no memory that is persistent between two cycles), and have no compass whatsoever (they are unable to agree on a common direction or orientation).

\paragraph{Related works}

The vast majority of literature on coordinated distributed robots considers that those robots are evolving in a \emph{continuous} two-dimentional Euclidian space and use visual sensors with perfect accuracy that permit to locate other robots with infinite precision~\cite{AFSY08c,SDY09j,SY99j,FPSW08j,DLP08j,DP09j}.

Several works investigate restricting the capabilities of both visibility sensors and motion actuators of the robots, in order to circumvent the many impossibility results that appear in the general continuous model. In~\cite{AOSY99j,FPSW05j}, robots visibility sensors are supposed to be accurate within a constant range, and sense nothing beyond this range. In~\cite{FPSW05j,BGT09c}, the space allowed for the motion actuator was reduced to a one-dimentional continuous one: a ring in~\cite{FPSW05j}, an infinite path in~\cite{BGT09c}.  
 
A recent trend was to shift from the classical continuous model to the \emph{discrete} model. In the discrete model, space is partitioned into a \emph{finite} number of locations. This setting is conveniently represented by a graph, where nodes represent locations that can be sensed, and where edges represent the possibility for a robot to move from one location to the other. Thus, the discrete model restricts both sensing and actuating capabilities of every robot. For each location, a robot is able to sense if the location is empty of if robots are positioned on it (instead of sensing the exact position of a robot). Also, a robot is not able to move from a position to another unless there is explicit indication to do so (\emph{i.e.}, the two locations are connected by an edge in the representing graph). The discrete model permits to simplify many robot protocols by reasoning on finite structures (\emph{i.e.}, graphs) rather than on infinite ones. However, as noted in most related papers~\cite{KMP08j,KKN08c,FIPS07c,FIPS08c}, this simplicity comes with the cost of extra symmetry possibilities, especially when the authorized paths are also symmetric (indeed, techniques to break formation such as those of~\cite{DLP08j} cannot be used in the discrete model). 

The two main problems that have been studied in the discrete robot model are gathering~\cite{KMP08j,KKN08c} and exploration~\cite{FIPS07c,FIPS08c}. For gathering, both breaking symmetry~\cite{KMP08j} and preserving symmetry are meaningful approaches. For exploration, the fact that robots need to stop after exploring all locations requires robots to ``remember'' how much of the graph was explored, \emph{i.e.}, be able to distinguish between various stages of the exploration process since robots have no persistent memory. As configurations can be distinguished only by robot positions, the main complexity measure is then the number of robots that are needed to explore a given graph. The vast number of symmetric situations induces a large number of required robots. For tree networks, \cite{FIPS08c} shows that $\Omega(n)$ robots are necessary for most $n$-sized tree, and that sublinear robot complexity is possible only if the maximum degree of the tree is $3$. In uniform rings, \cite{FIPS07c} proves that the necessary and sufficient number of robots is $\Theta(\log n)$, although it is required that the number $k$ of robots and the size  $n$ of the ring are coprime. Note that all previous approaches in the discrete model are \emph{deterministic}, \emph{i.e.}, if a robot is presented twice the same situation, its behavior is the same in both cases.

\paragraph{Our contribution}
In this paper, we initiate research on \emph{probabilistic} bounds and solutions in the discrete robot model, and focus on the \emph{exploration} problem of anonymous unoriented rings of any size. 
By contrast with~\cite{FIPS07c} while in the same system setting, we show that \emph{four} identical probabilistic robots are necessary and sufficient to solve the same problem, also removing the coprime constraint between the number of robots and the size of the ring. Our negative result show that for any ring of size at least four, there cannot exist any protocol with three robots in our setting, even if they are allowed to make use of probabilistic primitives. Our positive results are constructive, as we present a randomized protocol with four robots for any ring of size more than eight.

\paragraph{Outline}
The remaining of the paper is divided as follows. Section~\ref{sec:model} presents the system model that we use throughout the paper. Section~\ref{sec:impossible} provides evidence that no three probabilistic robots can explore every ring, while Section~\ref{sec:possible} presents our protocol with four robots. Section~\ref{sec:conclusion} gives some concluding remarks.

%% file: model.tex
\section{Model}
\label{sec:model}

\paragraph{Distributed System}

We consider systems of autonomous mobile entities called {\em agents} or {\em robots} 
evolving into a {\em graph}. We assume that the graph is a \emph{ring} of $n$ nodes, 
$u_0$,\dots, $u_{n-1}$, {\em i.e.}, $u_i$ is connected
to both $u_{i-1}$ and $u_{i+1}$ --- every computation over indices is assumed to be 
modulus $n$. The indices are used for notation purposes only: the nodes are {\em anonymous} 
and the ring is {\em unoriented}, {\em i.e.}, given two neighboring nodes $u$, $v$, there is
no kind of explicit or implicit labelling allowing to determine whether $u$ is on the right
or on the left of $v$. Operating in the ring are $k \leq n$ anonymous robots. 

A \emph{protocol} is a collection of $k$ \emph{programs}, one operating on each robot. 
The program of a robot consists in executing {\em Look\mbox{-}Compute\mbox{-}Move cycles} infinitely many times. 
That is, the robot first observes its environment (Look phase).
Based on its observation, a robot then (probabilistically or deterministically) decides --- according to its program --- to move or 
stay idle (Compute phase). 
When an robot decides a move, it moves to its destination during the Move phase. 

The robots do not communicate in an explicit way; however they see the position of 
the other robots and can acquire knowledge from this information. 
We assume that the robots cannot remember any previous observation nor computation performed 
in any previous step.  Such robots are said to be \emph{oblivious} (or \emph{memoryless}). 
The robots are also \emph{uniform} and \emph{anonymous}, i.e, they 
all have the same program using no local parameter (such that an identity) 
allowing to differentiate any of them.

\paragraph{Computations}

Time is represented by an infinite sequence of instants 0, 1, 2, \dots\  
At every instant $t \geq 0$, a non-empty subset of robots is activated to execute a cycle. 
The execution of each cycle is assumed to be \emph{atomic}: every robot that is 
activated at instant $t$ instantaneously executes a full cycle between $t$ and $t+1$.
Atomicity guarantees that at any instant the robots are on some nodes of 
the ring but not on edges. Hence, during a Look phase, a robot sees no robot on edges. 

We assume that during the Look phase, every robot can perceive 
whether several robots are located on the same node or not.  This ability is called \emph{Multiplicity Detection}.
We shall indicate by $d_i(t)$ the multiplicity of robots present in node $u_i$ at instant $t$.
More precisely $d_i(t) = j$ indicates that there is $j$ robots in node $u_i$ at instant $t$. If $d_i(t) \geq 2$, then we say that there is a {\em tower} in $u_i$ at instant $t$ (or simply there is a {\em tower} in $u_i$ when it is clear from the context). We say a node $u_i$ is {\em free at instant $t$} (or simply {\em free} when it is clear from the context) if $d_i(t) = 0$. Conversely, we say that $u_i$ is {\em occupied at instant $t$} (or simply {\em occupied} when it is clear from the context) if $d_i(t) \neq 0$.

Given an arbitrary orientation of the ring and a node $u_i$, $\gamma^{+i}(t)$ (respectively, $\gamma^{-i}(t)$) denotes 
the sequence $\langle d_i(t)d_{i+1}(t)\dots d_{i+n-1}(t)\rangle$
(resp., $\langle d_i(t)d_{i-1}(t)\dots d_{i-(n-1)}(t)\rangle$). 
The sequence $\gamma^{-i}(t)$ is called {\em mirror} of $\gamma^{+i}(t)$ and conversely.
Since the ring is unoriented, agreement on only one of the two sequences $\gamma^{+i}(t)$ and $\gamma^{-i}(t)$) is impossible. 
The (unordered) pair $\{\gamma^{+i}(t),\gamma^{-i}(t)\}$ is called the {\em view} of node $u_i$ at instant $t$ 
(we omit ``at instant $t$'' when it clear from the context). The view of $u_i$ is said to be {\em symmetric} if and only if 
$\gamma^{+i}(t) = \gamma^{-i}(t)$.  Otherwise, the view of $u_i$ is said to be {\em asymmetric}.

By convention, we state that the {\em configuration} of the system at instant $t$ is $\gamma^{+0}(t)$. 
Any configuration from which there is a probability 0 that a robot moves is said \emph{terminal}. 
Let $\gamma = \langle x_0x_1\dots x_{n-1}\rangle$ be a configuration. 
The configuration $\langle x_ix_{i+1}\dots x_{i+n-1}\rangle$ is obtained 
by rotating $\gamma$ of $i \in [0\dots n-1]$. 
Two configurations $\gamma$ and $\gamma^\prime$ are said {\em undistinguable} if and only if $\gamma^\prime$ can be obtained by rotating $\gamma$ or its mirror. Two configurations that are not undistinguable are said {\em distinguable}.
We designate by {\em initial configurations} the configurations 
from which the system can start at instant 0.

During the Look phase of some cycle, it may happen that both edges incident to a 
node $v$ currently occupied by the robot look identical in the snapshot, {\em i.e.}, 
$v$ lies on a symmetric axis of the configuration. In this case, if the robot decides to move, 
it may traverse any of the two edges. We assume the worst case decision in such cases, {\em i.e.}, 
that the decision to traverse one of these two edges is taken by an adversary.

We call \emph{computation} any infinite sequence of configurations $\gamma_0, \dots, \gamma_t, \gamma_{t+1}, \dots$ 
such that (1) $\gamma_0$ is a possible initial configuration and (2) 
for every instant $t \geq 0$, $\gamma_{t+1}$ is obtained from $\gamma_t$ after some 
robots (at least one) execute a cycle. Any transition $\gamma_t, \gamma_{t+1}$ is called a step of the 
computation. A computation $c$ \emph{terminates} if $c$ contains a terminal configuration. 

A \emph{scheduler} is a predicate on computations, that is, a scheduler define a
set of \emph{admissible} computations, such that every computation in this set satisfies
the scheduler predicate. Here we assume a \emph{distributed fair} scheduler. Distributed means that, at every 
instant, any non-empty subset of robots can be activated.
Fair means that every robot is activated infinitively often during a computation. A particular case of distributed fair scheduler is the \emph{sequential} fair scheduler: at every instant, one robot is activated and  every robot is activated infinitively often during a computation. In the following, we call sequential computation any computation that satisfies the sequential fair scheduler predicate. 

\paragraph{Problem to be solved}

We consider the {\em exploration} problem, where $k$ robots collectively explore a $n$-sized ring before stopping moving forever.
More formally, a protocol $\mathcal P$ \emph{deterministically} (resp. \emph{probabilistically}) solves the exploration problem if and only if every computation $c$ of $\mathcal P$ starting from a \emph{towerless configuration} satisfies:
\begin{enumerate}
\item $c$ terminates in \emph{finite time} (resp. with \emph{expected finite time}).
\item Every node is visited by at least one robot during $c$. 
\end{enumerate}

The previous definition implies that every initial configuration of the system in the problem we consider is {\em towerless}.

Using probabilistic solutions, termination is not certain, however the overall probability of non-terminating computations is 0.

%% file: imp.tex
\section{Negative Result}
\label{sec:impossible}

In this section, we show that the exploration problem is impossible to solve in our settings ({\em i.e.}, oblivious robots, anonymous ring, distributed scheduler, \dots) if there is less than four robots,  even in a probabilistic manner (Corollary \ref{cor:mesfesses}). The proof is made in two steps: 
\begin{itemize}
\item The first step is based on the fact that obliviousness constraints any exploration protocol to construct an implicit memory using the configurations. We show that if the scheduler behaves sequentially, then in any case except one, it is not possible to particularize enough configurations to memorize which nodes have been visited (Theorem \ref{theo:1} and Lemma \ref{lem:3n4}). 
\item The second step consists in excluding the last case (Theorem \ref{theo:2}).
\end{itemize}

\bigskip

\noindent Lemmas \ref{lem:nk1} to \ref{lem:carac} proven below are technical results that lead to Corollary \ref{coro}. The latter exhibits the minimal size of a subset of particular configurations required to solve the exploration problem.

\begin{definition}[MRP] Let $s$ be a sequence of configurations. The \emph{minimal relevant prefix} of $s$, noted ${\cal MRP}(s)$, is the maximal subsequence of $s$ where no two consecutive configurations are identical.
\end{definition}

\begin{lemma}\label{lem:nk1}Let ${\cal P}$ be any (probabilistic or deterministic) exploration protocol for $k$ robots in a ring of $n$ nodes. For every sequential computation $c$ of ${\cal P}$ that terminates, we have $|{\cal MRP}(c)| \geq n-k+1$. 
\end{lemma}
\begin{proof}
Let $c$ be a sequential computation that terminates. In the initial configuration of $c$ exactly $k$ nodes are already visited because there is at most one robot in each node. So, $n-k$ nodes are dynamically visited before $c$ terminates. As the computation is sequential, the computation contains at least $n-k+1$ different configurations: the initial one plus one configuration per node to be dynamically visited. Hence, $|{\cal MRP}(c)| \geq n-k+1$.
\end{proof}

\begin{lemma}\label{lem:nk1wt}Let ${\cal P}$ be any (probabilistic or deterministic) exploration protocol for $k$ robots in a ring of $n>k$ nodes. For every sequential computation $c$ of ${\cal P}$ that terminates, ${\cal MRP}(c)$ has at least $n-k+1$ configurations containing a tower.
\end{lemma}
\begin{proof}
Assume, by the contradiction, that there is a sequential computation $c$ of ${\cal P}$ that terminates and such that ${\cal MRP}(c)$ has less than $n-k+1$ configurations containing a tower. 

There exists a suffix $c'$ of $c$ starting from a configuration $\alpha$ without tower followed a suffix $s$ that only contains configurations with a tower. As $\alpha$ is a configuration without tower, $c'$ is an admissible sequential computation of ${\cal P}$. Moreover, as $c$ terminates, $c'$ terminates too. Hence, $|{\cal MRP}(c')| = n-k+1$ by Lemma \ref{lem:nk1} and all robots must be visited before $c'$ reaches its terminal configuration. As a consequence, $c'$ contains exactly $n-k$ steps of the form $\beta\beta'$ with $\beta \neq \beta'$. Now, the first of these steps in $c'$ is a step where one robot moves to a node already occupied by another robot (remember that the computation is sequential and the first step in ${\cal MRP}(c')$ is a step from a configuration without tower to a configuration with a tower). Hence, $c'$ contains at most $n-k-1$ steps where a new node is visited: $c'$ terminates before all robots are visited, a contradiction.
\end{proof}

\begin{lemma}\label{lem:nk1wt-k}Let ${\cal P}$ be any (probabilistic or deterministic) exploration protocol for $k$ robots in a ring of $n>k$ nodes. For every sequential computation $c$ of ${\cal P}$ that terminates, ${\cal MRP}(c)$ has at least $n-k+1$ configurations containing a tower of less than $k$ robots.
\end{lemma}
\begin{proof}
Assume, by the contradiction, that there is a sequential computation $c$ of ${\cal P}$ that terminates and such that ${\cal MRP}(c)$ has less than $n-k+1$ configurations containing a tower of less than $k$ robots. 

There exists a suffix $c'$ of $c$ starting from a configuration $\alpha$ without tower followed a suffix $s$ that only contains configurations with a tower. As $\alpha$ is a configuration without tower, $c'$ is an admissible sequential computation of ${\cal P}$. Moreover, as $c$ terminates, $c'$ terminates too. Hence, ${\cal MRP}(c')$ is constituted of a configuration with no tower followed by at least $n-k+1$ configurations containing a tower by Lemma \ref{lem:nk1wt} and all robots must be visited before $c'$ reaches its terminal configuration. 

As the first configuration of $c'$ is without tower, for every configuration $\alpha$ of ${\cal MRP}(c')$ with a tower there exists a unique step in ${\cal MRP}(c')$ of the form $\alpha'\alpha$ with $\alpha' \neq \alpha$. Now, as $c'$ is sequential, for each of these steps, if $\alpha$ contains a tower of $k$ robots, then no new node is visited during $\alpha' \neq \alpha$. By contradiction assumption, there is less than $n-k+1$ of steps $\beta'\beta$ such that $\beta$ contains a tower of less than $k$ robots. Moreover, no node is visited during the first of these steps (remember that the computation is sequential and the first of these steps is a step from a configuration without tower to a configuration with a tower). Hence, less that $n+k$ steps allow to dynamically visit new nodes in $c'$ and, as $c'$ is sequential, $c'$ terminates before all robots are visited, a contradiction.
\end{proof}

\begin{lemma}\label{lem:carac}Let ${\cal P}$ be any (probabilistic or deterministic) exploration protocol for $k$ robots in a ring of $n>k$ nodes. For every sequential computation $c$ of ${\cal P}$ that terminates, ${\cal MRP}(c)$ has at least $n-k+1$ configurations containing a tower of less than $k$ robots and any two of them are distinguable.
\end{lemma}
\begin{proof}
Consider any sequential computation $c$ of ${\cal P}$ that terminates.

By Lemma \ref{lem:nk1wt-k}, ${\cal MRP}(c)$ has $x$ configurations containing a tower of less than $k$ robots where $x \geq n-k+1$.

We first show that (**) {\em if $c$ contains at least two different configurations that are undistinguable, then there exists a sequential computation $c'$ that terminates and such that ${\cal MRP}(c')$ has $x'$ configurations containing a tower of less than $k$ robots where $x' < x$.} Assume that there two different undistinguable configurations $\gamma$ and $\gamma'$ in $c$ having a tower of less than $k$ robots. Without loss of generality, assume that $\gamma$ occurs at time $t$ in $c$ and $\gamma'$ occurs at time $t'>t$ in $c$. Consider the two following case:
\begin{enumerate}
\item {\bf $\gamma'$ can be obtained by applying a rotation of $i$ to $\gamma$.} Let $p$ be the prefix of $c$ from instant $0$ to instant $t$. Let $s$ be the suffix of $c$ starting at instant $t'+1$. Let $s'$ be the sequence obtained by applying a rotation of $-i$ to the configurations of $s$. As the ring and the robots are anonymous, $ps'$ is an admissible sequential computation that terminates. Moreover, by construction ${\cal MRP}(ps')$ has $x'$ configurations containing a tower of less than $k$ robots where $x' < x$. Hence (**) is verified in this case.
\item {\bf $\gamma'$ can be obtained by applying a rotation of $i$ to the mirror of $\gamma$.} We can prove (**) in this case by slightly modifying the proof of the previous case: we have just to apply the rotation of $-i$ to the \emph{mirrors} of the configurations of $s$.
\end{enumerate}

By (**), if ${\cal MRP}(c)$ contains less than $n-k+1$ distinguable configurations with a tower of less than $k$ robots, it is possible to (recursively) construct an admissible computation $c'$ of ${\cal P}$ such that ${\cal MRP}(c')$ has less than $n-k+1$ configurations containing a tower of less than $k$ robots, a contradiction to Lemma \ref{lem:nk1wt-k}. Hence, the lemma holds.
\end{proof}

\noindent From Lemma \ref{lem:carac}, we can deduce the following corollary:

\begin{corollary}\label{coro}
Considering any (probabilistic or deterministic) exploration protocol for $k$ robots in a ring of $n>k$ nodes, there exists a subset ${\cal S}$ of at least $n-k+1$ configurations such that:
\begin{enumerate}
\item Any two different configurations in ${\cal S}$ are distinguable, and
\item In every configuration in ${\cal S}$, there is a tower of less than $k$ robots.
\end{enumerate}
\end{corollary}

\begin{theorem}\label{theo:1}
$\forall k, 0 \leq k < 3, \forall n>k$, there is no exploration protocol (even probabilistic) of a $n$-size ring with $k$ robots.
\end{theorem}
\begin{proof}
First, for $k = 0$, the theorem is trivially verified. Consider then the case $k = 1$ and $k = 2$: with one robot it is impossible to construct a configuration with one tower; with two robots it is impossible to construct a configuration with one tower of less than $k$ robots ($k = 2$). Hence, for $k = 1$ and $k = 2$, the theorem is a direct consequence of Corollary \ref{coro}. 
\end{proof}

\begin{lemma}\label{lem:3n4}
$\forall n > 4$, there is no exploration protocol (even probabilistic) of a $n$-size ring with three robots.
\end{lemma}
\begin{proof}
With three robots, the size of the maximal set of distinguable configurations containing a tower of less than three robots is $\lfloor n/2 \rfloor$. By Corollary \ref{coro}, we have then the following inequality: $$\lfloor n/2 \rfloor \geq n-k+1$$From this inequality, we can deduce that $n$ must be less of equal than four and we are done.   
\end{proof}

\noindent From this point on, we know that, assuming $k<4$, Corollary \ref{coro} prevents the existence of any exploration protocol in any case except one: $k = 3$ and $n = 4$ (Theorem \ref{theo:1} and Lemma \ref{lem:3n4}). Actually, assuming that the scheduler is sequential is no sufficient to show the impossibility in this latter case: Indeed, there an exploration protocol for $k = 3$ and $n = 4$ if we assume a sequential scheduler. The protocol works as shown in Figure \ref{fig:algok3n4}. 

\begin{figure}[htpb]
\center
\includegraphics[scale=0.4]{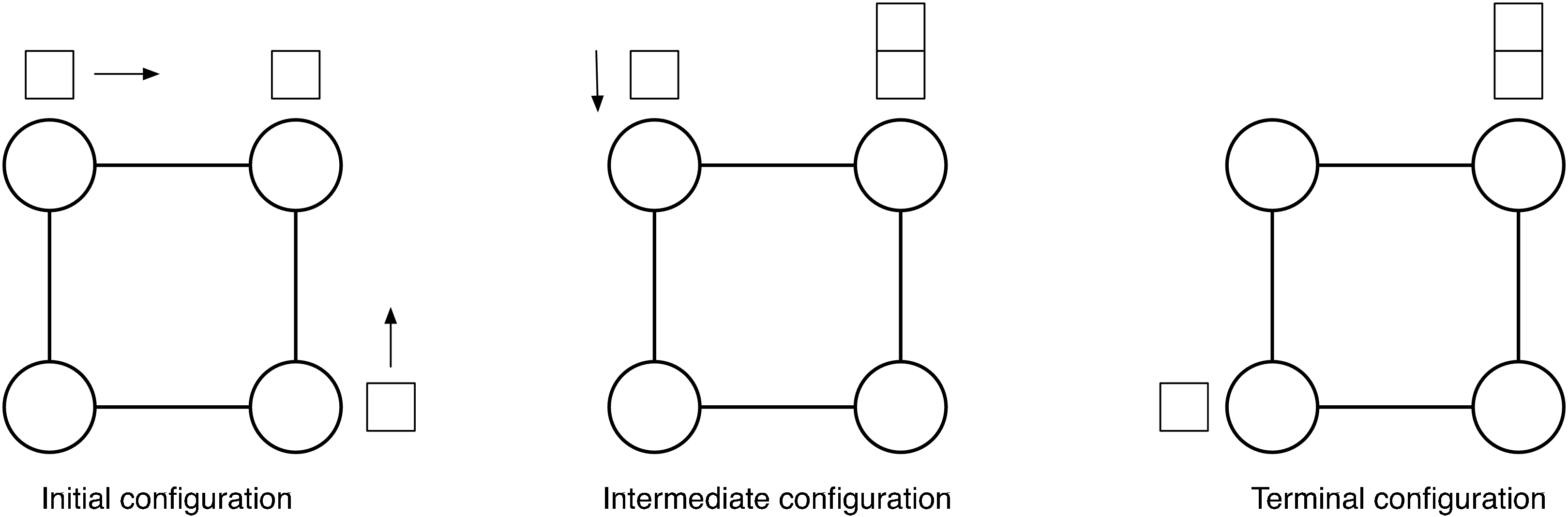}
\caption{Protocol for $n=4$ and $k=3$. (The arrows show the destinations of the robots if they are activated.)\label{fig:algok3n4}}
\end{figure}

We now show the impossibility in this latter using a (non-sequential) distributed scheduler. This proof is established by enumerating and testing all possible protocols for $k=3$ and $n=4$.

\begin{theorem}\label{theo:2}
There is no exploration protocol (even probabilistic) of a $n$-size ring with three robots for every $n > 3$.
\end{theorem}
\begin{proof}
Lemma \ref{lem:3n4} excludes the existence of any exploration protocol for three robots in a ring of $n>4$ nodes. Hence, to show this theorem, we just have to show that there is no exploration protocol for three robots working in a ring of four nodes.

Assume, by the contradiction, that there exists an exploration protocol ${\cal P}$ for three robots in a ring of four nodes. Then, any possible initial configuration is undistinguable with the configuration presented in Figure \ref{fig:init}. Moreover, any possible terminal configuration contains a tower and so is undistinguable with one of the three configurations presented in Figure \ref{fig:term}.

\begin{figure}[htpb]
\center
\includegraphics[scale=0.4]{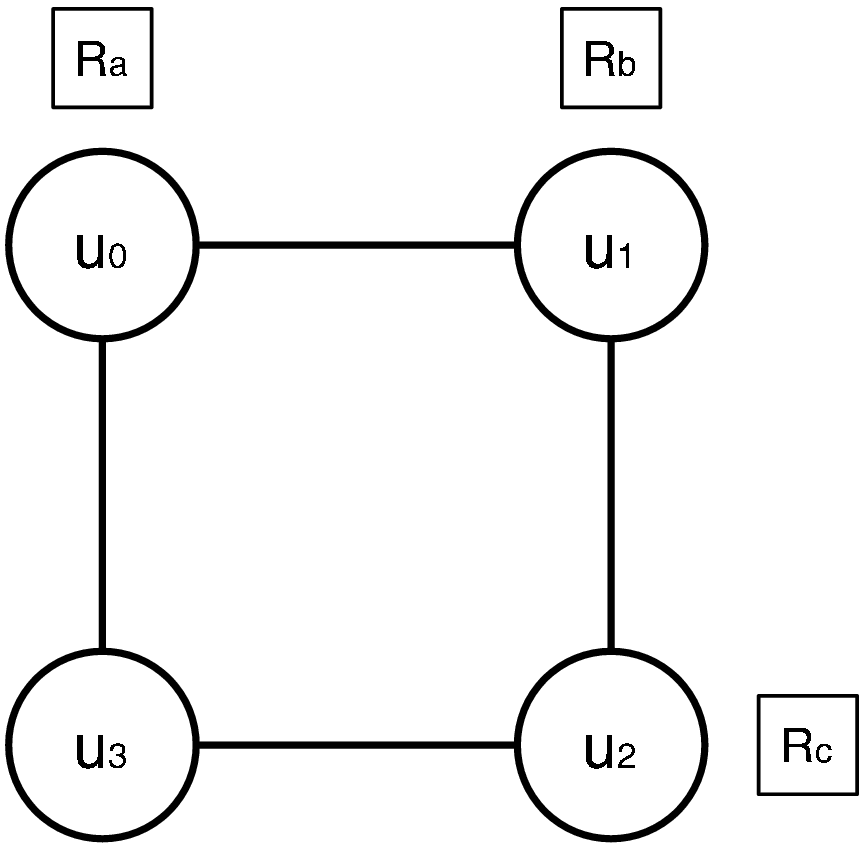}
\caption{Initial configuration for $n = 4$ and $k = 3$. (The indices are used for notation purposes only.)\label{fig:init}}
\end{figure}

\begin{figure}[htpb]
\center
\includegraphics[scale=0.45]{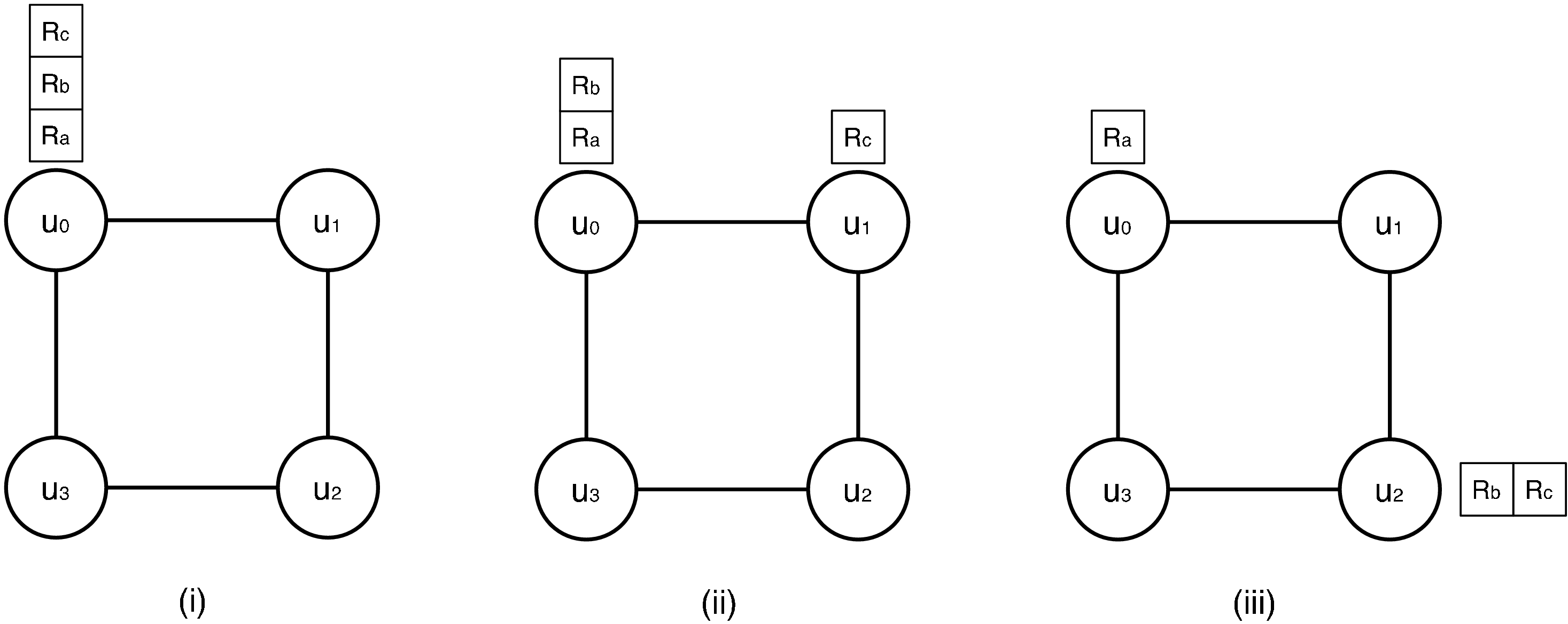}
\caption{Terminal configurations for $n = 4$ and $k = 3$. (The indices are used for notation purposes only.)\label{fig:term}}
\end{figure}

Consider that the system is initially in the configuration of Figure \ref{fig:init}. Three cases are possible at instant 0 using ${\cal P}$:

\begin{itemize}
\item {\em There is a strictly positive probability that robot $R_a$ (resp. robot $R_c$) moves to node $u_3$ if activated by the scheduler.\footnote{If ${\cal P}$ is deterministic the probability is 1 and if activated, $R_a$ moves in one step.}} In this case, assume that the scheduler activates $R_a$ until it moves. The probability that $R_a$ eventually moves is 1 (resp. $R_a$  moves in one step if ${\cal P}$ is determistic). Once $R_a$ has moved, $R_b$ has a strictly positive probability to move to node $u_0$ if activated. Assume then that the scheduler activates $R_b$ until it moves. The probability that $R_b$ eventually moves is 1. Repeating this scheme for $R_c$ and so on, it is possible to construct a distributed fair computation that does not terminate in finite expected time (resp. in finite time, if if ${\cal P}$ is determistic), a contradiction.
\item {\em There is a strictly positive probability that robot $R_a$ (resp. robot $R_c$) moves to node $u_1$ if activated by the scheduler.} In this case, there an admissible computation where $R_a$ and $R_c$ moves to node $u_1$ in the first step. At instant 1, the system is in a configuration that is undistinguable with configuration (i) of Figure \ref{fig:term}. As node $u_3$ is still not visited in this case, any configuration that is undistinguable with configuration (i) cannot be terminal. There is also an admissible computation where only $R_a$ moves to node $u_1$ in the first step. At instant 1, the system is in a configuration that is undistinguable with configuration (ii) of Figure \ref{fig:term}. As node $u_3$ is still not visited in this case, any configuration that is undistinguable with configuration (ii) cannot be terminal. Moreover, assuming that the system reaches a configuration undistinguable from configuration (i) of Figure \ref{fig:term} at instant 1, there is a strictly positive probability that the three robots moves (the configuration is not terminal and all robots have the same view). If they move, the adversary can choose which incident edge they traverse because the configuration is symmetric. Hence, we can obtain a configuration undistinguable with configuration (iii) of Figure \ref{fig:term} and where node $u_3$ is still not visited. Thus, any configuration that is undistinguable with configuration (iii) cannot be terminal. Hence, no configuration can be terminal, a contradiction.
\item {\em There is a strictly positive probability that robot $R_b$ moves if activated by the scheduler.} Assume that the scheduler activates $R_b$ until it moves. Then, the probability that $R_b$ eventually moves is 1. Once $R_b$ decide to move, the adversary can choose the edge that $R_b$ traverses because the view from $R_b$ is symmetric. Hence, the system can reache the configuration $\gamma$: $R_a$ is in node $u_0$, $R_b$ and $R_c$ and in node $u_2$. This configuration is undistinguable with configuration (iii) in Figure \ref{fig:term} and node $u_3$ is still not visited. Consider the two following cases:
\begin{itemize}
\item {\em The probability that $R_a$ moves, if activated, is 0.} Then, there is a strictly positive probability that $R_c$ (resp. $R_b$) moves if activated. Assume that the scheduler activates $R_a$ and then $R_c$ until $R_c$ moves. The probability that $R_c$ eventually moves is 1 and as the the view from $R_c$ is symmetric, the adversary can decide which edge $R_c$ will traverse. Assume that the adversary forces $R_c$ to go to node $u_1$, the system reaches a configuration undistinguable with the initial configuration. Repeating the same scheme infinitively often, we obtain a distributed fair computation that does not terminate in finite expected time, a contradiction.
\item {\em The probability that $R_a$ moves ,if activated, is strictly positive.} Assume that the scheduler activates $R_a$ until it moves. Then, the probability that $R_a$ eventually moves is 1 and as the the view from $R_a$ is symmetric, the adversary can decide which edge $R_a$ will traverse. Assume that $R_a$ moves to node $u_1$, the system reaches the following configuration: $R_a$ is in node $u_1$, $R_b$ and $R_c$ are in node $u_2$, and node $u_3$ is still not visited. This configuration is undistinguable with configuration (ii) in Figure \ref{fig:term}. Consider the two following cases:
\begin{itemize}
\item {\em The probability that $R_c$ (resp. $R_b$) moves, if activated, is strictly positive.} 
\begin{itemize}
\item {\em Assume that the destination of $R_c$, if $R_c$, is node $u_3$.} Then, the system reaches a configuration undistinguable from initial configuration. Repeating same the scheme infinitively often, we obtain a distributed fair computation that does not terminate in finite expected time, a contradiction.
\item {\em Assume that the destination of $R_c$, if $R_c$ moves, is node $u_1$.} Then, the destination of $R_b$, if $R_b$ moves, is node $u_1$ too. Hence, there is an admissible computation where $R_b$ and $R_c$ move to node $u_1$. In this case, the system reaches a configuration that is not distinguable from configuration (i) in Figure \ref{fig:term} while node $u_3$ is still not visited. In this case, no configuration can be terminal, a contradiction.
\end{itemize}
\item {\em The probability that $R_b$ (resp. $R_c$) moves, if activated, is 0.} Then, the probability that $R_a$ moves is strictly positive. Consider the two following cases:
\begin{itemize}
\item {\em Assume that the destination of $R_a$, if $R_a$, is node $u_2$.} In this case, there is an admissible computation where $R_a$ move to node $u_2$: the system reaches a configuration that is not distinguable from configuration (i) in Figure \ref{fig:term} while node $u_3$ is still not visited. In this case, no configuration can be terminal, a contradiction. 
\item {\em Assume that the destination of $R_a$, if $R_a$, is node $u_0$.}  Assume that the scheduler activates $R_b$, $R_c$, and then $R_a$ until $R_a$ moves. The probability that $R_a$ eventually moves is 1 and we retreive a configuration that is undistinguable with configuration $\gamma$. Repeating the same scheme infinitively often, we obtain a fair distributed computation that does not terminate in finite expected time, a contradiction.
\end{itemize} 
\end{itemize}
\end{itemize}
\end{itemize}
In all cases, we obtain a contradiction: there no exploration protocol for three robots in a ring of $n>4$ nodes and the theorem is proven.
\end{proof}

\noindent From Theorems \ref{theo:1} and \ref{theo:2}, we can deduce the following corollary:

\begin{corollary}\label{cor:mesfesses}
$\forall k, 0 \leq k < 4, \forall n > k$, there is no exploration protocol (even probabilistic) of a $n$-size ring with $k$ robots.
\end{corollary}

%% file: algo.tex
\section{Positive Result}
\label{sec:possible}

In this section, we propose a probabilistic exploration protocol for $\boldsymbol{k = 4}$ robots in a ring of $\boldsymbol{n > 8}$ nodes. We first define some useful terms in Subsection \ref{sub:def}. We then give the general principle of the protocol in Subsection \ref{sub:prin}. Finally, we fully describe and prove the protocol in Subsection \ref{sub:des}.

\subsection{Definitions}\label{sub:def} 

Below, we define some terms that characterize the configurations.

We call {\em segment} any maximal non-empty elementary path of occupied nodes. The {\em length of a segment} is the number of nodes that compose it. We call {\em $x$\mbox{-}segment} any segment of length $x$. An {\em isolated node} is a node belonging to a $1$\mbox{-}segment. 

We call {\em hole} any maximal non-empty elementary path of free nodes. The {\em length of a hole} is the number of nodes that compose it. We call {\em $x$\mbox{-}hole} any hole of length $x$. In the hole $h = u_i, \dots, u_{k}$  ($k \geq i$) the nodes $u_i$ and $u_{k}$ are terms as the {\em extremities} of $h$. We call \emph{neighbor} of an hole any node that does not belong to the hole but is neighbor of one of its extremities. In this case, we also say that the hole is a \emph{neighboring hole} of the node. By extension, any robot that is located at a neighboring node of a hole is also referred to as a neighbor of the hole.

We call {\em arrow} a maximal elementary path $u_i, \dots, u_{k}$ of length at least four such that $(i)$ $u_i$ and $u_{k}$ are occupied by one robot, $(ii)$ $\forall j \in [i+1 \dots k-2]$, $u_j$ is free, and $(iii)$ there is a tower of two robots in $u_{k-1}$. The node $u_i$ is called the {\em arrow tail} and the node $u_{k}$ is called the {\em arrow head}. The {\em size} of an arrow is the number of free nodes that compose it, {\em i.e.}, its the length of the arrow path minus 3. Note that the minimal size of an arrow is $1$ and the maximal size is $n-4$. Note also that when there is an arrow in a configuration, the arrow is unique. An arrow is said {\em primary} if its size is 1. An arrow is said {\em final} if its size is $n-4$.

\begin{figure}[htpb]
\center
\includegraphics[scale=0.3]{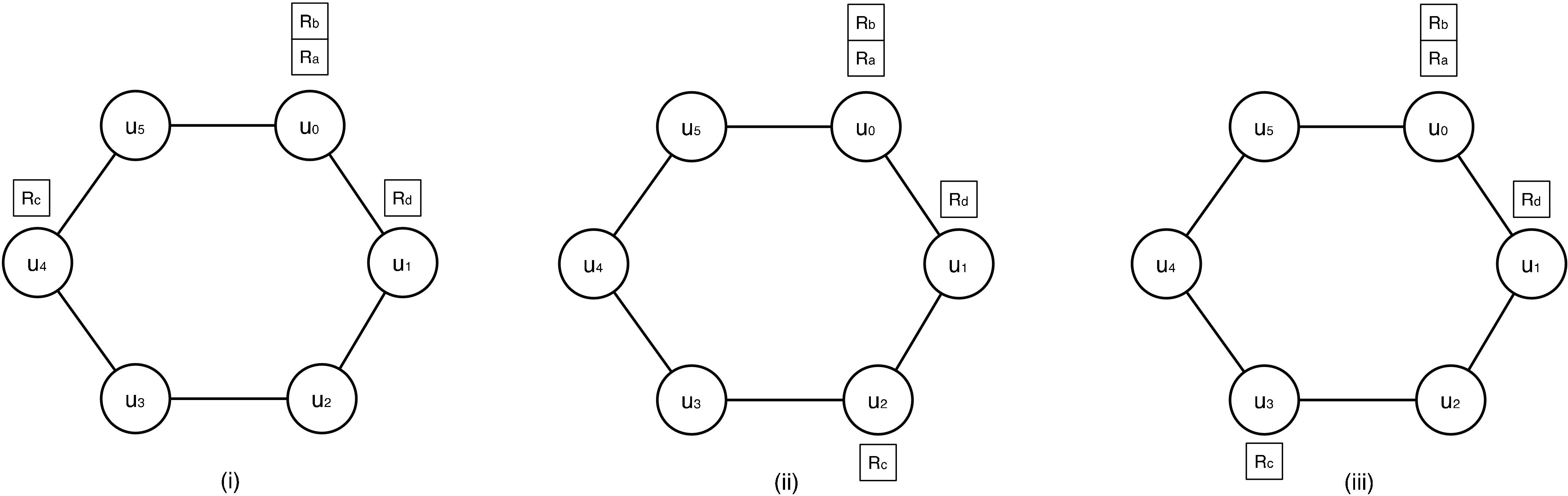}
\caption{Arrows\label{fig:arrows}}
\end{figure}

Figure \ref{fig:arrows} illustrates the notion of arrows: In Configuration $(i)$ the arrow is formed by the path $u_4$, $u_5$, $u_0$, $u_1$; the arrow is primary; the node $u_4$ is the tail and the node $u_1$ is the head. In Configuration $(ii)$, there is a final arrow (the path $u_2$, $u_3$, $u_4$, $u_5$, $u_0$, $u_1$). Finally, the size of the arrow in Configuration $(iii)$ (the path $u_3$, $u_4$, $u_5$, $u_0$, $u_1$) is 2.    

\subsection{Overview of the solution}\label{sub:prin}

Our protocol (Algorithm \ref{algo:main}) proceeds in three distinct phases: 
\begin{itemize} 
\item {\bf Phase $\mathrm{I}$:}
 Starting from a configuration without tower, the robots move along the ring in such a way that ($i$) they 
 never form any tower and ($2$) form a unique segment (a $4$-segment) in finite expected time.
\item {\bf Phase $\mathrm{II}$:}
 Starting from a configuration with a unique segment, the four robots form an primary arrow in finite expected time. 
The $4$\mbox{-}segment is maintained until the primary arrow is formed.
\item {\bf Phase $\mathrm{III}$:} 
 Starting from a configuration where the four robots form a primary arrow, the arrow tail moves toward the arrow head in such way that the existence of an arrow is always maintained. The protocol terminates when robots form a final arrow. At the termination, all nodes have been visited.
\end{itemize}

Note that the protocol we propose is probabilistic. As a matter of fact, as most as possible the robots move deterministically. However, we use randomization to break the symmetry in some cases: When the system is in a symmetric configuration, the scheduler may choose synchronously to activated some processes in such way that the system stays in a symmetric configuration. To break the symmetry despite the choice of the scheduler, we proceed as follows: The activated nodes toss a coin (with a uniform probability) during their Compute phase. If they win the toss, they decide to move, otherwise they decide to stay idle. In this case, we say that the robots {\bf try to move}. Conversely, when a process deterministically decides to move in its Compute phase, we simply say that the process {\bf moves}.

\begin{algorithm*}[htpb]
\scriptsize
\begin{algorithmic}[1]
\If{the four robots do not form a final arrow}\label{firstline}
\If{the configuration contains neither an arrow nor a $4$-segment}
\State Execute Procedure $Phase\ \mathrm{I}$;
\Else
\If{the configuration contains a $4$-segment} 
\State Execute Procedure $Phase\ \mathrm{II}$;
\Else \Rem{the configuration contains an arrow}
\State Execute Procedure $Phase\ \mathrm{III}$;
\EndIf
\EndIf
\EndIf
\end{algorithmic}
\caption{The protocol. \label{algo:main}}
\end{algorithm*}

\subsection{Detailed description of the solution}\label{sub:des}

\input{phase1}

\input{phase2}

\input{phase3}

%% file: phase1.tex
\subsubsection{Phase $\mathrm{I}$}   

Phase $\mathrm{I}$ is described in Algorithm \ref{algo:P1}. The aim of this phase is to eventually form a $4$-segment without creating any tower during the process. Roughly speaking, in asymmetric configurations, robots moves determiniscally  (Lines \ref{line:1}, \ref{line:2}, \ref{line:5}, \ref{line:6}). By contrast, in symmetric configurations, robots moves probabilistically using \TRY (Lines \ref{line:3} and \ref{line:4}). Note that in all case, we prevent the tower formation by applying the following constraint: a robot can move through a neighboring hole $\mathcal H$ only if its length is at least 2 or if the other neighboring robot can move through $\mathcal H$. 

\begin{algorithm*}[htpb]
\scriptsize
\begin{algorithmic}[1]
\If{the configuration contains a $3$-segment}
\If{I am the isolated robot}
\State \MOVE toward the $3$-segment through the shortest hole;\label{line:1}
\EndIf

\Else

\If{the configuration contains a unique $2$-segment} \Rem{ Two robots are isolated}

\If{I am at the closest distance from the $2$-segment}
\State \MOVE toward the $2$-segment through the hole having me and an extremity of the $2$-segment as neighbors;\label{line:2}
\EndIf

\Else

\If{the configuration contains (exactly) two $2$-segments} 

\If{I am a neighbor of a longuest hole}
\State \TRY toward the other $2$-segment through my neighboring hole;\label{line:3}
\EndIf

\Else \Rem{the four robots are isolated}

\State Let $l_{max}$ be the length of the longuest hole;
\If{every robot is neighbor of a $l_{max}$\mbox{-}hole}
\State \TRY through a neighboring $l_{max}$\mbox{-}hole;\label{line:4}
\Else

\If{3 robots are neighbors of a $l_{max}$\mbox{-}hole}
\If{I am neighbor of only one $l_{max}$\mbox{-}hole}
\State \MOVE toward the robot that is neighbor of no $l_{max}$\mbox{-}hole through my shortest neighboring hole;\label{line:5}
\EndIf

\Else \Rem{2 robots are neighbors of the unique $l_{max}$\mbox{-}hole}

\If{I am neighbor of the unique $l_{max}$\mbox{-}hole}
\State \MOVE through my shortest neighboring hole;\label{line:6}
\EndIf
\EndIf

\EndIf
\EndIf
\EndIf
\EndIf
\end{algorithmic}
\caption{Procedure $Phase\ \mathrm{I}$. \label{algo:P1}}
\end{algorithm*}

The following lemma (Lemma \ref{lem:notower}) shows that no tower can by creating during Phase $\mathrm{I}$. The next one (Lemma \ref{lem:p1}) shows that executing Algorithm \ref{algo:P1}, a $4$-segment is eventually created.

\begin{lemma}\label{lem:notower} If the configuration at instant $t$ contains neither a $4$\mbox{-}segment nor a tower, then the configuration at instant $t+1$ contains no tower.
\end{lemma}
\begin{proof} Let $\gamma$ be the configuration at instant $t$. First, note that the robots executes $Phase\ \mathrm{I}$ (Algorithm \ref{algo:P1}) in $\gamma$. Note also that $\gamma$ satisfies one of the following cases:
\begin{itemize}
\item \emph{$\gamma$ contains a $3$\mbox{-}segment}. In this case, only the (unique) isolated robot can move and, if it does, it moves to a free node (see Line \ref{line:1}). Hence, no tower is created at instant $t+1$. 
\item \emph{$\gamma$ contains a unique $2$\mbox{-}segment}. Two cases are possible:
\begin{itemize}
\item \emph{There is a unique isolated robot $\mathcal R$ at the closest distance from the $2$\mbox{-}segment}. In this case, only $\mathcal R$ can move and, if it does, it moves to free node (see Line \ref{line:2}), so no tower is created at instant $t+1$.
\item \emph{The two isolated robots are at the same distance from the $2$\mbox{-}segment}. In this case, the two isolated robots can move but as they follow their shortest path to the $2$\mbox{-}segment (see Line \ref{line:2}) and there is no tower in $\gamma$, they follow distinct paths and no tower is created at instant $t+1$.
\end{itemize}
Hence, in the two subcases no tower is created at instant $t+1$.
\item \emph{$\gamma$ contains two $2$\mbox{-}segments}. In this case, as there is four robots and the size of the ring is greater than 8, the size of the longuest hole is at least three. In such a configuration, the only possible moves are the moves where robots move through one of their neighboring holes of length at least two (see Line \ref{line:3}). Hence, all moving robots move to a different free node: no tower is created at instant $t+1$.
\item \emph{$\gamma$ contains four isolated robots}. Let $l_{max}$ be the length of the longuest hole in $\gamma$. In this case, as there is four robots and the size of the ring $n$ is greater than 8, $l_{max}\geq 2$. Consider then the following three subcases:
\begin{itemize}
\item \emph{Every robot is neighbor of a $l_{max}$\mbox{-}hole}. In this case, the configuration is symmetric. Every robot can move in the next step but to a neighboring hole of size at least two (see Line \ref{line:4}). So, all moving robots move to a different free node. Hence, no tower is created at instant $t+1$.
\item \emph{Three robots are neighbors of a $l_{max}$\mbox{-}hole}. Let $\mathcal R$ be the robot that is not neighbor of any $l_{max}$\mbox{-}hole. In this case, the robots that may move (at most two) move through their neighboring hole having $\mathcal R$ as other neighbor (see Line \ref{line:5}). As $\mathcal R$ cannot move, no tower is created at instant $t+1$.
\item \emph{Two robots, say $\mathcal R_1$ and $\mathcal R_2$, are neighbors of the unique $l_{max}$\mbox{-}hole}. In this case, only $\mathcal R_1$ and $\mathcal R_2$ can move. If $\mathcal R_1$ (resp. $\mathcal R_2$) moves, then $\mathcal R_1$ (resp. $\mathcal R_2$) through its neighboring hole having not $\mathcal R_2$ (resp. $\mathcal R_1$) as other neighbor (see Line \ref{line:6}). So, all moving robots move to a different free node. As a consequence, no tower is created at instant $t+1$.
\end{itemize}
\end{itemize}
In all cases, the configuration obtained at instant $t+1$ contains no tower and the lemma holds. 
\end{proof}

\begin{lemma}\label{lem:p1} Starting from any initial configuration, the system reaches in finite expected time a configuration containing a $4$\mbox{-}segment.
\end{lemma}
\begin{proof}
Any initial configuration contains no tower. If the initial configuration contains a $4$\mbox{-}segment, the lemma trivially holds. Consider now the case where the initial configuration contains neither a $4$\mbox{-}segment nor a tower.

By Lemma \ref{lem:notower}, while the system does not reaches a configuration containing a $4$\mbox{-}segment, the system remains in configurations containing no tower. For a given $n$-size ring network, the number of such configuration is \emph{finite}. So, to prove the lemma, we have just to show that from any configuration containing neither a $4$\mbox{-}segment nor a tower, there is always a strictly positive probability that the system eventually reaches a configuration containing a $4$\mbox{-}segment (despite the choices of the scheduler). To see this, consider a configuration $\gamma$ containing neither a $4$\mbox{-}segment nor a tower and split the study into the following cases:
\begin{enumerate}
\item\label{one} \emph{$\gamma$ contains a $3$\mbox{-}segment}. In this case, only the unique isolated robot can move and by the fairness property, it eventually does: it moves toward the $3$\mbox{-}segment through the shortest hole (see Line \ref{line:1}). So, until the system reaches a configuration containing a $4$\mbox{-}segment, only the isolated robot moves and at each move the length of the shortest hole decreases. Hence, the system reaches a configuration containing a $4$\mbox{-}segment \emph{in finite time}.
\item\label{unique} \emph{$\gamma$ contains a unique $2$\mbox{-}segment}. Following the same scheme as in the previous case, we can see that the system reaches a configuration containing a $4$\mbox{-}segment \emph{in finite time}.
\item\label{two} \emph{$\gamma$ contains two $2$\mbox{-}segments}. In this case, the robots that are neighbors of a longuest hole (at least two) can \emph{try} to move (see Line \ref{line:3}). So, by fairness property, a non-empty set of these robots, say $S$, is eventually activated by the scheduler. Now, every robot in $S$ decides with a uniform probability to move or not. So, there a strictly positive probability that only one robot in $S$ decides to move. In this case, we retreive the previous case and we are done.
\item\label{onelmax} \emph{$\gamma$ contains four isolated nodes}. Let $l_{max}$ be the length of the longuest hole in $\gamma$. Let study the following subcases:
\begin{enumerate}
\item\label{twolmax} \emph{Only two robots are neighbors of a $l_{max}$\mbox{-}hole}. In this case, the two robots that are neighbors of the unique $l_{max}$\mbox{-}hole can move. So, by fairness property, either one or both of them eventually move through their shortest neighboring hole  (see Line \ref{line:6}). After such moves, either $(i)$ the system is still in a configuration containing four isolated nodes and where two robots are neighbors of a unique longuest hole but the size of the longuest hole increased, or $(ii)$ the system is in a configuration containing a unique $2$\mbox{-}segment, or $(iii)$ the system is in a configuration containing two $2$\mbox{-}segments. Hence, the system reaches in finite time a configuration satisfying $(ii)$ or $(iii)$, {\em i.e.}, we eventually retreive the cases \ref{unique} or \ref{two}, and we are done. 
\item\label{threelmax} \emph{Exactly three robots are neighbors of a $l_{max}$\mbox{-}hole}. Let $\mathcal R_0$ be the robot that is not neighbor of a $l_{max}$\mbox{-}hole.  Let $\mathcal R_1$ and $\mathcal R_2$ be the two robots that are neighbor of exactly one $l_{max}$\mbox{-}hole. In this case, only $\mathcal R_1$ and $\mathcal R_2$ can move (see Line \ref{line:5}) and by fairness property at least one of them eventually does. If only one of them moves, then we retreive Subcase 4.(a) or Case \ref{unique}, and we are done. If both $\mathcal R_1$ and $\mathcal R_2$ move, then the system reaches $(i)$ either a configuration where exactly three robots are neighbors of a longuest hole of length $l_{max}+1$ or $(ii)$ a configuration containing a $3$\mbox{-}segment. In Case $(i)$, if we repeat the argument, we can see that we eventually retreive Subcase 4.(a), Case \ref{one}, or Case \ref{unique}, and we are done. In Case $(ii)$, we directly retreive Case \ref{one} and we are done.
\item \emph{The four robots are neighbors of a $l_{max}$\mbox{-}hole}. In this case, the configuration is symmetric and all robots \emph{try} move (see Line \ref{line:6}). Now, despite the choice of the scheduler, there is a strictly positive probability that only one robot probabilistically decides to move. In this case, the robot moves through one of its neighboring $l_{max}$\mbox{-}hole of size at least two (to provide the tower creation). As a consequence, we retreive Subcases 4.(a) or 4.(b) and we are done.
\end{enumerate}
\end{enumerate}
Hence, in all cases there is a strictly positive probability that the system eventually reaches a configuration containing a $4$\mbox{-}segment from $\gamma$ and the lemma holds.
\end{proof}

%% file: phase2.tex
\subsubsection{Phase $\mathrm{II}$}

Phase $\mathrm{II}$ is described in Algorithm \ref{algo:P2}: Starting from a configuration where there is a $4$-segment on nodes $u_i,u_{i+1},u_{i+2},u_{i+3}$, the system eventually reaches a configuration where a primary arrow is formed on nodes $u_i,u_{i+1},u_{i+2},u_{i+3}$. To that goal, we proceed as follows: Let $\mathcal R_1$ and $\mathcal R_2$ be the robots located at the nodes $u_{i+1}$ and $u_{i+2}$ of the $4$-segment. $\mathcal R_1$ and $\mathcal R_2$ try to move to $u_{i+2}$ and $u_{i+1}$, respectively. Eventually only one of these robots moves and we are done, as proven in the two next lemmas.

\begin{algorithm*}
\scriptsize
\begin{algorithmic}[1] 
\If{I am not located at an extremity of the $4$-segment}
\State \TRY toward my neighboring node that is not an extremity of the $4$-segment;
\EndIf
\end{algorithmic}
\caption{Procedure $Phase\ \mathrm{II}$. \label{algo:P2}}
\end{algorithm*}

\begin{lemma}\label{segment}
Let $\gamma$ be a configuration containing a $4$\mbox{-}segment $u_i,u_{i+1},u_{i+2},u_{i+3}$. If $\gamma$ is the configuration at instant $t$, then the configuration at instant $t+1$ is either identical to $\gamma$ or the configuration containing the primary arrow $u_i,u_{i+1},u_{i+2},u_{i+3}$.
\end{lemma}
\begin{proof}
Let $\mathcal R_1$ (resp. $\mathcal R_2$) be the robot located at node $u_{i+1}$ (resp. $u_{i+2}$) in $\gamma$. In $\gamma$, all robots executes Algorithm \ref{algo:P2} (see Algorithm \ref{algo:main}). So, from $\gamma$, only $\mathcal R_1$ and $\mathcal R_2$ can move: $\mathcal R_1$ can move to node $u_{i+2}$ and $\mathcal R_2$ can move to node $u_{i+1}$ (see Algorithm \ref{algo:P2}). When one or both of these robots, we obtain a configuration containing either a $4$\mbox{-}segment or a primary arrow in $u_i,u_{i+1},u_{i+2},u_{i+3}$ and the lemma holds.
\end{proof}

\begin{lemma}\label{lem:p2} From a configuration containing a $4$\mbox{-}segment, the system reaches a configuration containing a primary arrow in finite expected time.
\end{lemma}
\begin{proof}
By Lemma \ref{segment}, we know that starting from a configuration $\gamma$ containing a $4$\mbox{-}segment, the system either remains in the same configuration or reaches a configuration containing a primary arrow. Let $\mathcal R_1$ and $\mathcal R_2$ be the robots that are not located at the extremity of the $4$\mbox{-}segment in $\gamma$. Only $\mathcal R_1$ and $\mathcal R_2$ can (probabilistically) decide to move in $\gamma$. Also, by the fairness property, eventually one or both of them are activated. Now, despite the choice of the scheduler, there is  a strictly positive probability that only one of them probabilistically decide to move: in this case, the system reaches a configuration containing a primary arrow  (see Algorithm \ref{algo:P2}) and we are done.
\end{proof}

%% file: phase3.tex
\subsubsection{Phase $\mathrm{III}$}

Phase $\mathrm{III}$ is described in Algorithm \ref{algo:P3}. This phase is fully deterministic: Let $\mathcal H$ be the hole between the tail and the head of arrow. The robot located at the arrow tail traverses $\mathcal H$. When it is done, the system is in a terminal configuration containing a final arrow: all nodes have been visited as shown is the theorem below.

\begin{algorithm*}[htpb]
\scriptsize
\begin{algorithmic}[1] 
\If{I am the arrow tail}
\State \MOVE toward the arrow head through the hole having me and the arrow head as neighbor;
\EndIf
\end{algorithmic}
\caption{Procedure $Phase\ \mathrm{III}$. \label{algo:P3}}
\end{algorithm*}

\begin{theorem}\label{theo:p3}
Algorithm \ref{algo:main} is a probabilistic exploration protocol for {\bf 4} robots in a ring of $\boldsymbol{n>8}$ nodes.
\end{theorem}
\begin{proof}
The proof of the theorem is based on the two following claims:
\begin{enumerate}
\item\label{term} {\em Any configuration containing a final arrow is terminal}.

{\bf Proof:} Immediate, see Line \ref{firstline} of Algorithm \ref{algo:main}.
\item\label{inc} {\em From a configuration containing a non-final arrow of length $x$, the system eventually reaches a configuration containing a $x+1$\mbox{-}arrow.} 

{\bf Proof:} In such a configuration, only the arrow tail can move. By the fairness property, the robot located at the arrow tail moves \emph{in finite time}: it moves through its neighboring hole having the arrow head as other neighbor (see Algorithm \ref{algo:P3}). As a consequence, the size of the arrow is incremented to $x+1$ and we are done.
\end{enumerate}
Using the two previous claims, we now prove the lemma in two step:
\begin{itemize}
\item {\bf Termination.} {\em Any computation of Algorithm \ref{algo:main} terminates in finite expected time.} 

{\bf Proof:} Immediate from Lemmas \ref{lem:p1}, \ref{lem:p2}, Claims \ref{term} and \ref{inc}.
\item {\bf Partial Correctness.} {\em When a computation of Algorithm \ref{algo:main} terminates, any node has been visited.} 

{\bf Proof:} By Lemma \ref{lem:p1}, starting from any initial configuration, the system reaches in finite expected time a configuration containing a $4$\mbox{-}segment say $u_i,u_{i+1},u_{i+2},u_{i+3}$. By Lemmas \ref{segment} and \ref{lem:p2}, from this configuration the system reaches in finite expected time a configuration containing an arrow on $u_i,u_{i+1},u_{i+2},u_{i+3}$. Hence, when the phase $\mathrm{III}$ starts, nodes $u_i$, $u_{i+1}$, $u_{i+2}$, and $u_{i+3}$ are already visited. By Claim \ref{inc}, the robots executes then Algorithm \ref{algo:P3} until the computation terminates. Let $\mathcal P$ be the path $u_{i-1},\dots,u_{i-n+4}$. By Claim \ref{inc}, until the computation terminated, only the robot located at the arrow tail can move and it it move following $\mathcal P$. Hence, when the computation terminates all nodes of $\mathcal P$ have been visited ({\em i.e.}, nodes $u_{i-1}$, \dots, $u_{i-n+4}$) and, as nodes $u_i$, $u_{i+1}$, $u_{i+2}$, $u_{i+3}$ have also been visited, we are done. 
\end{itemize}
\end{proof}

%% file: ccl.tex
\section{Conclusion}
\label{sec:conclusion}

We provided evidence that for the exploration problem in uniform rings, randomization could shift complexity from $\Theta(\log n)$ to $\Theta(1)$. While applying randomization to other problem instances is an interesting topic for further research, we would like to point out immediate open questions raised by our work:

\begin{enumerate}

\item Though we were able to provide a general algorithm for any $n$ (strictly) greater than eight, it seems that ad hoc solutions have to be designed when $n$ is between five and eight (included). 

\item Our protocol is optimal with respect to the number of robots. However, the efficiency (in terms of exploring time) is only proved to be finite. Actually computing the convergence time from our proof argument is feasible, but it would be more interesting to study how the number of robots relates to the time complexity of exploration, as it seems natural that more robots will explore the ring faster.

\end{enumerate}